\documentclass{article}
\usepackage{fullpage,graphicx,amssymb,amsmath,cite,color}

\newcommand{\reals}{\ensuremath{\mathbb{R}}}
\newcommand{\mst}{\operatorname{MST}}
\newcommand{\opt}{\operatorname{OPT}}
\newcommand{\pr}{\operatorname{Pr}}
\newcommand{\dist}{\operatorname{dist}}
\newcommand{\inter}{\operatorname{inter}}
\newcommand{\adj}{\operatorname{Adj}}

\newcommand{\tp}{\hspace{10pt}}

\DeclareMathOperator*{\argmin}{arg\,min}
\DeclareMathOperator*{\dif}{dif}

\newtheorem{theorem}{Theorem}
\newtheorem{lemma}[theorem]{Lemma}
\newtheorem{cor}[theorem]{Corollary}

\newtheorem{prop}[theorem]{Proposition}

\newtheorem{definition}{Definition}
\newenvironment{proof}{{\noindent\bf Proof.}}{$\Box$\newline}

\begin{document}

\title{Bounding Interference in Wireless Ad Hoc Networks
with Nodes in Random 
Position\thanks{This work was supported in part by the Natural Sciences
and Engineering Research Council of Canada (NSERC).}}

\author{
Majid Khabbazian\thanks{%
University of Winnipeg, Winnipeg, Canada, {\tt m.khabbazian@uwinnipeg.ca}}
\and Stephane Durocher\thanks{%
University of Manitoba, Winnipeg, Canada, {\tt durocher@cs.umanitoba.ca}}
\and Alireza Haghnegahdar\thanks{%
University of British Columbia, Vancouver, Canada,
{\tt alirezah@ece.ubc.ca}}
}

\date{October 28, 2011}

\maketitle

\begin{abstract}
The interference at a wireless node $s$ can be modelled
by the number of wireless nodes whose transmission ranges cover $s$. 
Given a set of positions for wireless nodes,
the interference minimization problem is to assign a
transmission radius (equivalently, a power level)
to each node such that the resulting communication graph
is connected, while minimizing the maximum interference.
We consider the model introduced by von Rickenback et al.\ (2005),
in which each transmission range is represented by a ball
and edges in the communication graph are symmetric.
The problem is NP-complete in two dimensions (Buchin 2008) and
no polynomial-time approximation algorithm is known.
Furthermore, even in one dimension (the highway model), 
the problem's complexity is unknown 
and the maximum interference of a set of $n$ wireless nodes
can be as high as $\Theta(\sqrt{n})$ (von Rickenback et al.\ 2005).
In this paper we show how to solve the problem efficiently in 
settings typical for wireless ad hoc networks.
In particular, 
we show that if node positions are 
represented by a set $P$ of $n$ points selected uniformly and independently at
random over a $d$-dimensional rectangular region, for any fixed $d$, 
then the topology given by the closure of the Euclidean minimum spanning 
tree of $P$ has maximum interference $O(\log n)$ with high probability.
We extend this bound to a general class of communication graphs
over a broad set of probability distributions.
Next we present a local algorithm that constructs a graph from this class;
this is the first local algorithm to provide an upper bound
on the expected maximum interference.
Finally, we discuss an empirical evaluation of our algorithm with
a suite of simulation results.

{\bf keywords:} interference, topology control, minimum spanning tree,
random distribution, expectation
\end{abstract}

%%% SECTION - INTRODUCTION
\section{Introduction}\label{sec:introduction}
%%%%%%%%%%%%%%%%%%%%%%%%%%%%
% Bounding the Maximum Interference of Random Sensors in the Plane
% SECTION - INTRODUCTION	
%
% file last modified on November 14, 2011
%%%%%%%%%%%%%%%%%%%%%%%%%%%%

%%% SUBSECTION
\subsection{Motivation}
\label{sec:introduction.motivation}
Establishing connectivity in a wireless network can be a complex task
for which various (sometimes conflicting) objectives may need to be optimized. 
To permit a packet to be routed from any origin node to any destination
node in the network, the corresponding communication graph must be connected 
(or strongly connected if unidirectional communication links are permitted).
In addition to requiring connectivity, various properties can be imposed
on the network,
including low power consumption \cite{wattenhoferTON2005,santi2005}, 
bounded average traffic load \cite{durocherJOIN2009,hyytia2006}, 
small average hop distance between sender-receiver pairs
\cite{benkert2008}, 
low dilation ($t$-spanner) 
\cite{benkert2008,smid2007,burkhart2004,li2002,kanj2008,bose2005,damian2006},
and minimal interference;
this latter objective, minimizing interference
(and, consequently, minimizing the required bandwidth),
is the focus of much recent research 
\cite{wattenhofer2005,wattenhofer2009,wattenhofer2008,wattenhofer2005b,%
halldorsson2008,buchin2008,tan2011,kranakis2010,korman2011,%
kowalski2010,bilo2008,sharma2009,santi2005,benkert2008}
and of this paper.

We adopt the interference model introduced by von Rickenbach et al.\ 
\cite{wattenhofer2005}
(see Section~\ref{sec:introduction.problemDefinition}).
We model transmission in a wireless network by assigning to each wireless 
node $p$ a radius of transmission $r(p)$, 
such that every node within distance $r(p)$ of $p$ can receive a 
transmission from $p$, whereas no node a greater distance from $p$ can.
Consequently, the interference at node $p$ is the number of nodes 
that have $p$ within their respective radii of transmission.
Given a set of wireless nodes whose positions are represented by a set of
points $P$, we consider the problem of identifying a connected
network on $P$ that minimizes the maximum interference.
The problem of constructing the network is equivalent to that of
assigning a transmission radius to each node.
That is, once the transmission radius of each node is fixed, the corresponding
communication graph and its associated maximum interference are also fixed.
Conversely, once a graph is fixed, the transmission radius of each node
is determined by the distance to its furthest neighbour.

Given a set of points $P$ in the plane, finding a connected graph on $P$
that minimizes the maximum interference is NP-complete \cite{buchin2008}.
A polynomial-time algorithm exists that returns a solution with
maximum interference $O(\sqrt{n})$, where $n = |P|$ \cite{halldorsson2008}.
Even in one dimension, for every $n$ there exists a set of $n$ points $P$ 
such thta any graph on $P$ has 
maximum interference $\Omega(\sqrt{n})$ \cite{wattenhofer2005}.
All such known examples involve specific constructions 
(i.e., exponential chains). 
We are interested in investigating a more realistic class of wireless
ad hoc networks:
those whose node positions observe common random distributions
that better model actual wireless ad hoc networks.

When nodes are positioned on a line (often called the {\em highway model}), 
a simple heuristic is to assign to each node a radius of transmission
that corresponds to the maximum of the distances to
its respective nearest neighbours to the left and right.
In the worst case, such a strategy can result in 
$\Theta(n)$ maximum interference
when an optimal solution has only $\Theta(\sqrt{n})$ maximum interference
\cite{wattenhofer2005}.
Recently, Kranakis et al.\ \cite{kranakis2010} 
showed that if $n$ nodes are positioned uniformly at random on an interval,
then the maximum interference provided by this heuristic is 
$\Theta(\sqrt{\log n})$ with high probability.

In this paper, 
we examine the corresponding problem in two and higher dimensions.
We generalize the nearest-neighbour path used in the highway
model to the Euclidean minimum spanning tree (MST),
and show that with high probability, the maximum interference 
of the MST of a set of $n$ points selected uniformly at random 
over a $d$-dimensional region $[0,1]^d$
is $O(\log n)$, for any fixed $d \geq 1$.
Our techniques differ significantly from those used 
by Kranakis et al.\ to achieve their results in one dimension.
As we show in Section~\ref{sec:bounds},
our results also apply to a broad class of random distributions,
denoted $\mathcal{D}$, that includes both the uniform random distribution and
realistic distributions for modelling 
random motion in mobile wireless networks,
as well as to a large class of connected spanning graphs that includes the MST.

In Section~\ref{sec:localAlg} we present 
a local algorithm that constructs a topology whose maximum interference
is $O(\log n)$ with high probability when node positions are selected
according to a distribution in $\mathcal{D}$.
Previous local algorithms for topology control 
(e.g., the cone-based local algorithm (CBTC) \cite{wattenhoferTON2005})
attempt to reduce transmission radii (i.e., power consumption), 
but not necessarily the maximum
interference. Although reducing transmission radii at many nodes
is often necessary to reduce the maximum interference, 
the two objectives differ;
specifically, some nodes may require large transmission
radii to minimize the maximum interference.
Ours is the first local algorithm to provide a non-trivial upper bound
on maximum interference.
Our algorithm can be applied to any existing topology to refine it
and further reduce its maximum interference. Consequently, our solution
can be used either independently, or paired with another topology control
strategy.
Finally, we discuss an empirical evaluation of our algorithm with
a suite of simulation results in Section~\ref{sec:simulation}.

%%% SUBSECTION
%\subsection{Problem Statement, Model, and Definitions}
\subsection{Model and Definitions}
\label{sec:introduction.problemDefinition}

We represent the position of a wireless node as a point in Euclidean space,
$\reals^d$, for some fixed\footnote{In the majority of instances, 
two or three dimensions suffice to model an actual wireless 
network. Our results are presented in terms of an arbitrary $d$
since this permits expressing a more general result
without increasing the complexity of the corresponding notation.}
$d \geq 1$.
For simplicity, we refer to each node by its corresponding point.
Similarly, we represent a wireless network by its communication graph,
a geometric graph whose vertices are a set of points $P \subseteq \reals^d$.
Given a (simple and undirected) graph $G$,
we employ standard graph-theoretic notation, 
where $V(G)$ denotes the vertex set of $G$ and $E(G)$ 
denotes
its edge set.
We say vertices $u$ and $v$ are {\em $k$-hop neighbours}
if there is a simple path of length $k$ from $u$ to $v$ in $G$.
When $k=1$ we say $u$ and $v$ are neighbours.

We assume a uniform range of communication for each node
and consider bidirectional communication links, 
each of which is represented by an undirected graph edge connecting two nodes. 
Specifically, each node $p$ has some {\em radius of transmission},
denoted by the function $r : P \to \reals^+$,
such that a node $q$ receives a transmission from $p$ 
if and only if $\dist(p,q) \leq r(p)$, where $\dist(p,q) = ||p-q||_2$ 
denotes the Euclidean distance between points $p$ and $q$ in $\reals^d$.
For simplicity, suppose each node has an infinite radius of reception, 
regardless of its radius of transmission.

\begin{definition}[Communication Graph]
\label{def:communicationGraph}
A graph $G$ is a {\em communication graph} 
with respect to a point set $P \subseteq \reals^d$ and 
a function $r : P \to \reals^+$ if
\begin{enumerate}
\item $V(G) = P$, and
\item for all vertices $p$ and $q$ in $V(G)$,
\begin{equation}
\label{eqn:communicationGraphEdge}
\hspace{-2pt}
\{p,q\} \in E(G) \Leftrightarrow \dist(p,q) \leq \min\{ r(p), r(q)\} .
\end{equation}
\end{enumerate}
\end{definition}

Together, set $P$ and function $r$ uniquely determine 
the corresponding communication graph $G$.
Alternatively, a communication graph can be defined as the closure
of a given embedded graph.
Specifically, if instead of being given $P$ and $r$, 
we are given an arbitrary graph $H$ embedded in $\reals^d$, 
then the set $P$ is trivially determined by $V(H)$
and a transmission radius for each node $p \in V(H)$ can be assigned
to satisfy \eqref{eqn:communicationGraphEdge} by
\begin{equation}
\label{eqn:transmissionRadiusAlt}
r(p) =  \max_{q \in \adj(p)} \dist(p,q) , 
\end{equation}
where $\adj(p) = \{ q \mid \{q,p\} \in E(H) \}$ 
denotes the set of vertices adjacent to $p$ in $H$.
The communication graph determined by $H$ is 
the unique edge-minimal supergraph of $H$ that satisfies 
Definition~\ref{def:communicationGraph}.
We denote this graph by $H'$ and refer to it as the {\em closure} of graph $H$.
Therefore, a communication graph $G$ can be defined either
as a function of a set of points $P$ 
and an associated mapping of transmission radii $r:P \to \reals^+$,
or as the closure of a given embedded graph $H$ (where $G = H'$).

\begin{definition}[Interference]
\label{def:interference}
Given a communication graph $G$ 
the {\em interference} at node $p$ in $V(G)$ is 
\[ \inter_G(p) = | \{ q \mid q \in V(G) \setminus \{ p \} 
\text{~and~} \dist(q,p) \leq r(q) \} |\]
and the {\em maximum interference} of $G$ is
\[ \inter(G) = \max_{p \in V(G)} \inter_G(p) . \]
\end{definition}
In other words, the interference at node $p$, 
denoted $\inter_G(p)$, is the number of nodes $q$ such that 
node $p$ lies within $q$'s radius of transmission.
This does not imply the existence of the edge $\{p,q\}$
in the corresponding communication graph; 
such an edges exists if and only if the relationship is reciprocal, i.e.,
$q$ also lies with $p$'s radius of transmission.

Given a point set $P$,
let $\mathcal{G}(P)$ denote the set of connected communication graphs on $P$.
Let $\opt(P)$ denote the optimal maximum interference attainable over graphs 
in $\mathcal{G}(P)$. That is,
\[ \opt(P) = \min_{G \in \mathcal{G}(P)} \inter(G) 
= \min_{G \in \mathcal{G}(P)} \max_{p \in V(G)} \inter_G(p) . \]

Thus, given a set of points $P$ representing the positions of wireless nodes,
the {\em interference minimization problem} is to find
a connected communication graph $G$ on $P$ that spans $P$ 
such that the maximum interference is minimized (i.e., its maximum interference
is $\opt(P)$).
In this paper we examine the maximum interference of the 
communication graph determined by the closure of $\mst(P)$, where 
$\mst(P)$ denotes the Euclidean minimum spanning tree of the point set $P$.
Our results apply with high probability, 
which refers to probability at least $1 - n^{-c}$,
where $n = |P|$ denotes the number of networks nodes
and $c \geq 1$ is fixed.

%%% SUBSECTION
%\subsection{Overview of Results}

%%% SECTION - RELATED WORK
\section{Related Work}\label{sec:relatedWork}
%%%%%%%%%%%%%%%%%%%%%%%%%%%%
% Bounding the Maximum Interference of Random Sensors in the Plane
% SECTION - RELATED WORK
%
% file last modified on July 29, 2011
%%%%%%%%%%%%%%%%%%%%%%%%%%%%

%%% SUBSECTION
\subsection{Bidirectional Interference Model}
In this paper we consider the bidirectional interference model
(defined in Section~\ref{sec:introduction.problemDefinition}).
This model was introduced by 
von Rickenback et al.\ \cite{wattenhofer2005},
who gave a polynomial-time approximation algorithm that finds a solution
with maximum interference $O(n^{1/4} \cdot \opt(P))$
for any given set of points $P$ on a line,
and a one-dimensional construction showing that $\opt(P) \in \Omega(\sqrt{n})$ 
in the worst case, where $n = |P|$.
Halld\'orsson and Tokuyama \cite{halldorsson2008}
gave a polynomial-time algorithm that returns a solution
with maximum interference $O(\sqrt{n})$ 
for any given set of $n$ points in the plane.
Buchin \cite{buchin2008} showed that finding an optimal solution 
(one whose maximum interference is exactly $\opt(P)$) is NP-com\-plete 
in the plane.
Tan et al.\ \cite{tan2011} gave an $O(n^3 n^{O(\opt(P))})$-time algorithm for 
finding an optimal solution for any given set of points $P$ on a line.
Kranakis et al.\ \cite{kranakis2010}
showed that for any set of
points $P$ selected uniformly at random from the unit interval,
the maximum interference of the nearest-neighbour path (MST) has maximum 
interference $\Theta(\sqrt{\log n})$ with high probability.
Finally, Sharma et al.\ \cite{sharma2009} consider heuristic solutions
to the two-dimensional problem.

%%% SUBSECTION
\subsection{Unidirectional Interference Model}
If communication links are not bidirectional (i.e., edges are directed)
and the communication graph is required to be strongly connected, 
then the worst-case maximum interference decreases.
Under this model, von Rickenback et al.\ \cite{wattenhofer2009} 
and Korman \cite{korman2011}
give polynomial-time algorithms that return solutions
with maximum interference $O(\log n)$
for any given set of points in the plane,
and a one-dimensional construction showing that in the worst case
$\opt(P) \in \Omega(\log n)$.

%%% SUBSECTION
\subsection{Minimizing Average Interference}
In addition to results that examine the problem of minimizing 
the maximum interference, some work has addressed the problem
of minimizing the average interference,
e.g., Tan et al.\ \cite{tan2011}
and Moscibroda and Wattenhofer \cite{wattenhofer2005b}.

%%% SECTION - BOUNDS
\section{Bounds}\label{sec:bounds}
%%%%%%%%%%%%%%%%%%%%%%%%%%%%
% Bounding the Maximum Interference of Random Sensors in the Plane
% SECTION - BOUNDS
%
% file last modified on November 2, 2011
%%%%%%%%%%%%%%%%%%%%%%%%%%%%

%%% SUBSECTION
\subsection{Generalizing One-Dimensional Solutions}
Before presenting our results on random sets of points,
we begin with a brief discussion regarding the possibility
of generalizing existing algorithms that provide approximate solutions
for one-dimensional instances of the interference minimization
problem (in an adversarial deterministic input setting).

Since the problem of identifying a graph that achieves the optimal
(minimum) interference is NP-hard in two or more dimensions
\cite{buchin2008},
it is natural to ask whether one can design a polynomial-time algorithm
to return a good approximate solution.
Although Rickenback et al.\ \cite{wattenhofer2005}
give a $\Theta(n^{1/4})$-approximate algorithm
in one dimension \cite{wattenhofer2005},
the current best polynomial-time algorithm in two (or more) dimensions
by Halld\'orsson and Tokuyama \cite{halldorsson2008}
returns a solution whose maximum interference is $O(\sqrt{n})$;
as noted by Halld\'orsson and Tokuyama,
this algorithm is not known to guarantee any approximation
factor better than the immediate bound of $O(\sqrt{n})$.
The algorithm of Rickenback et al.\ uses two strategies for constructing
respective communication graphs, and returns the graph with the lower
maximum interference; an elegant argument
that depends on Lemma~\ref{lem:wattenhofer2005} bounds the resulting
worst-case maximum interference by $\Theta(n^{1/4} \cdot \opt(P))$.
The two strategies correspond roughly to a) $\mst(P)'$ and
b) classifying every $\sqrt{n}$th node as a hub,
joining each hub to its left and right neighbouring hubs
to form a network backbone,
and connecting each remaining node to its closest hub.
The algorithm of Halld\'orsson and Tokuyama applies $\epsilon$-nets, resulting
in a strategy that is loosely analogous to a generalization
of the hub strategy of Rickenback et al.\ to higher dimensions.
One might wonder whether the hybrid approach of Rickenback et al.\
might be applicable in higher dimensions.
Specifically, can a good approximation factor be guaranteed
by returning the better of the respective graphs returned
by the $\epsilon$-net algorithm of Halld\'orsson and Tokuyama
and the communication graph determined by $\mst(P)'$?
To apply this idea directly in two or more dimensions
would require generalizing the following property established by
von Rickenback et al.:
\begin{lemma}[von Rickenback et al.\ \cite{wattenhofer2005} (2005)]
\label{lem:wattenhofer2005}
For any set of points $P \subseteq \reals$,
\[ \opt(P) \in \Omega\left(\sqrt{\inter(\mst(P)')}\right) . \]
\end{lemma}
However, von Rickenback et al.\ also show that 
for any $n$, there exists a set of $n$ points $P \subseteq \reals^2$
such that $\opt(P) \in O(1)$ and $\inter(\mst(P)') \in \Theta(n)$,
which implies that Lemma~\ref{lem:wattenhofer2005}
does not hold in higher dimensions.
Consequently, techniques such as those used by von Rickenback et al.\
to bound the approximation factor of their algorithm in one dimension 
do not immediately generalize to higher dimensions.

%%% SUBSECTION
\subsection{Randomized Point Sets}
Although using the hybrid approach
of von Rickenback et al.\ \cite{wattenhofer2005}
directly may not be possible,
Kranakis et al.\ \cite{kranakis2010} recently
showed that if a set $P$ of $n$ points is selected uniformly at random from
an interval, then the maximum interference of the communication graph
determined by $\mst(P)'$ is $\Theta(\sqrt{\log n})$ with high probability.
Throughout this section, we assume general position of points;
specifically, we assume
that the distance between each pair of nodes is unique.
This can be expressed formally as
$\forall \{p_1,p_2,q_1,q_2\} \subseteq P$,
\[
  \dist(p_1,q_1)=\dist(p_2,q_2) \Leftrightarrow \{p_1,q_1\}=\{p_2,q_2\}.
\]
We begin by introducing the following definitions:

\begin{definition}[Primitive Edge]
\label{def:primitiveEdge}
An edge $\{p,q\}\in E(G)$ in a communication graph $G$ is {\em primitive} 
if $\min\{r(p),$ $r(q) \} = \dist(p,q)$.
\end{definition}

\begin{definition}[Bridge]
\label{def:bridge}
An edge $\{p,q\}\in E(G)$ in a communication graph $G$ is {\em bridged}
if there is a path joining $p$ and $q$ in $G$ consisting of at most three edges,
each of which is of length less than $\dist(p,q)$.
\end{definition}

\begin{definition}[$\mathcal{T}(P)$]
\label{def:T}
Given a set of points $P$ in $\reals^d$,
$\mathcal{T}(P)$ is the set of all communication graphs $G$ with $V(G)=P$
such that no primitive edge $\{p,q\}\in E(G)$ is bridged.
\end{definition}

Let $\mathcal{C}(R,r,d)$ be the minimum number of $d$-dimensional
balls of radius $r$ required to cover a $d$-dimensional ball of radius $R$.
The following property follows
since $\reals^d$ is a doubling metric space for any constant $d$
\cite{heinonen2001}
(equivalently, $\reals^d$
and has constant doubling dimension \cite{gupta2003,fraigniaud2006}):

\begin{prop}
\label{prop:doubling}
If $d\in \Theta(1)$ and $R/r\in \Theta(1)$,
then $\mathcal{C}(R,r,d) \in \Theta(1)$.
\end{prop}

We now bound the maximum interference of any graph in $\mathcal{T}(P)$.

\begin{theorem}
\label{thm:log}
  Let $P$ be a set of points in $\reals^d$.
  For any graph $G\in\mathcal{T}(P)$,
  \[
    \inter(G) \in O \left( \log \left(
    \frac{d_{\max}(G)}{d_{\min}(G)} \right) \right) ,
  \]
where $d_{\max}(G) =\max_{\{s,t\}\in E(G)} \dist(s,t)$
and $d_{\min}(G) = \min_{\{s,t\}\in E(G)} \dist(s,t)$.
\end{theorem}

\begin{proof}
We first normalize the scale of $P$ to simplify the proof.
Let $Q = \{ p \cdot \alpha \mid p \in P\}$ denote a uniform scaling of $P$ 
by a factor of $\alpha = 1/d_{\min}(G)$ and
let $H$ denote the corresponding communication graph.
That is, $\{u,v\} \in E(G) 
\Leftrightarrow \{ u \cdot \alpha , v \cdot \alpha \} \in E(H)$.
Similarly, scale transmission radii such that 
each node's transmission radius in $Q$ 
is $\alpha$ times its corresponding node's transmission radius in $P$.
Thus, 
\begin{equation}
\label{eqn:transformation}
d_{\min}(H) = 1 \qquad \text{and} \qquad
d_{\max}(H) = \frac{d_{\max}(G)}{d_{\min}(G)} .
\end{equation}
We say an edge $\{q_1,q_2\}\in E(H)$ causes interference at a node $p$
if $p$ is within the transmission range of either $q_1$ or $q_2$.
Let $p$ be a node in $V(H)$ that has interference $\inter(H)$.
Let $E(p) \subseteq E(H)$ be the set of all primitive edges
that cause interference at $p$.
Since there are $\inter(H)$ nodes whose transmission ranges cover $p$,
we get that $|E(p)| \geq \inter(H) / 2$.
That is, there are at least $\inter(H)/2$ primitive edges that cause
interference at node $p$.
Therefore, to prove the theorem it suffices to show that
\begin{equation}
    |E(p)| \in O ( \log ( d_{\max}(H) ) ) .
\label{eqn:log}
\end{equation}

Let $g=\lceil \log (d_{\max}(H)) \rceil$.
Partition $E(p)$ into $g+1$ subsets, $E_0, E_1,\ldots, E_g$, such that
for each $0\leq i\leq g$, $E_i$ is the set of all edges in $E(p)$
whose length is in $[2^i , 2^{i+1} )$.
Since %$d_{\min}(H) = 2^0 d_{\min}(H)$ and
$d_{\max}(H) \leq 2^g$, it follows that
\[
E(p) = \bigcup_{0 \leq i \leq g} E_i
\qquad \text{and} \qquad
\forall i \neq j, \ E_i \cap E_j = \varnothing .
\]
We now show that $|E_i|\in O(1)$ for every $i$, $0\leq i\leq g$,
from which \eqref{eqn:log} follows immediately.

For each integer $i$, $0\leq i\leq g$, let $V_i$ be the set of all
nodes in $V(H)$ that are incident to an edge in $E_i$
and let $V'_i\subseteq V_i$ be the set of nodes in $V_i$
that have $p$ in their transmission radii.
By our assumption of general position,
there is an injective function from the set of primitive edges
in $E_i$ to nodes in $V'_i$, giving that
\begin{equation}
|V_i| \geq |V_i'| \geq |E_i|.
\label{eqn:log2}
\end{equation}
By definition of $E_i$, $V_i$, and $V_i'$,
every node in $V'_i$ is contained in the ball with centre $p$
and radius $2^{i+1}$.
Furthermore, every node $v$ in $V_i$ is contained in the ball with centre $p$
and radius $2^{i+2}$, because either $v \in V'_i$
or $v$ is adjacent to a node $w$ in $V'_i$;
thus, $\dist(p,v) \leq \dist(p,w)+\dist(w,v) \leq 2 \cdot 2^{i+1}$.
By Proposition~\ref{prop:doubling},
for a constant dimension $d$,
  $\mathcal{C}(2^{i+1}, 2^{i-2},d)\in O(1)$ and
  $\mathcal{C}(2^{i+2}, 2^{i-2},d)\in O(1)$.
Suppose $|E_i|\not\in O(1)$.
Hence by \eqref{eqn:log2},
$|E_i|$, $|V_i|$, and $|V'_i|$ are each $\omega(1)$.
In particular, for a sufficiently large point set,
\begin{equation}
    |V'_i|\geq \mathcal{C}(2^{i+1}, 2^{i-2},d)\cdot
    \left[\mathcal{C}(2^{i+2}, 2^{i-2},d)+1\right].
\label{eqn:log3}
\end{equation}
Any ball of radius $2^{i+1}$ can be covered with
$\mathcal{C}(2^{i+1}, 2^{i-2},d)$ balls of radius $2^{i-2}$.
Therefore, by \eqref{eqn:log3} and the pigeonhole principle,
there must be a ball $B_i$ of
radius $2^{i-2}$ that contains a set of
nodes $V''_i$, such that $V''_i\subseteq V'_i$ and
$|V_i''| \geq \mathcal{C}(2^{i+2}, 2^{i-2},d)+1$.
Let $W_i$ be the set of nodes in $V_i$ that are adjacent to some node
in $V''_i$ by some edge in $E_i$.
Since the length of every edge in $E_i$ is at least $2^{i}$
and the ball $B_i$ has radius $2^{i-2}$,
every node in $W_i$ must lie outside $B_i$. Thus,
  \begin{equation}
  \label{equ:empty}
    W_i \cap V''_i = \varnothing.
  \end{equation}
We consider two cases: i) there is a node $q$ in $W_i$ that is adjacent
to at least two nodes in $V''_i$ by edges in $E_i$,
and ii) every node in $W_i$ is adjacent to only one node in $V''_i$
by some edge in $E_i$, i.e., $|W_i|\geq|V''_i|$.

{\em Case i.}
Let $p_1$ and $p_2$ denote two nodes in $V''_i$
such that edges $\{p_1,q\}$ and $\{p_2,q\}$ are in $E_i$.
Without loss of generality, assume that $\dist(p_1,q)>\dist(p_2,q)$
(by our general position assumption).
Consider the path $\langle p_1,p_2,q \rangle$ from $p_1$ to $q$.
This path has two edges.
Also, $\dist(p_2,q)<\dist(p_1,q)$ and $\dist(p_1,p_2)<\dist(p_1,q)$,
because $\dist(p_1,p_2)\leq 2^{i-1}$ (as $p_1$ and $p_2$ are within
a ball of radius $2^{i-2}$) and $\dist(p_1,p_2)\geq 2^{i}$
(as the edge $\{p_1,q\}$ is in $E_i$).
Since $\{p_1,q\}$ is a primitive edge in $H$ and $H\in\mathcal{T}(Q)$,
$\{p_1, q\}$ cannot be bridged, deriving a contradiction.

{\em Case ii.} We have $|W_i| \geq|V''_i| \geq
\mathcal{C}(2^{i+2}, 2^{i-2},d)+1$.
Since every node in $W_i$ lies in a ball of radius $2^{i+2}$
(as $W_i\subseteq V_i$), and a ball of radius $2^{i+2}$ can be covered
with $\mathcal{C}(2^{i+2}, 2^{i-2},d)+1$ balls of
radius $2^{i-2}$, there must be a ball of radius
$2^{i-2}$ that contains at least two nodes $q_1$ and $q_2$
from $W_i$.
By~\eqref{equ:empty}, $W_i\cap V''_i=\varnothing$.
By definition, there must be two edges in $E_i$ that connect $q_1$ and $q_2$
to two distinct nodes $p_1$ and $p_2$ in $V''_i$.
Without loss of generality, assume that $\dist(p_1,q_1)>\dist(p_2,q_2)$.
The length of the edge $\{p_1,q_1\}$ is greater than those of
$\{p_1,p_2\}$, $\{p_2,q_2\}$ and $\{q_2,q_1\}$, because
$\dist(p_1,p_2)\leq 2^{i-1}$,
$\dist(q_1,q_2)\leq 2^{i-1}$,
and $\dist(p_1,q_1)\geq 2^{i}$ (as $\{p_1,q_1\}\in E_i$).
Therefore, every edge of the path of length three 
$\langle p_1,p_2,q_2,q_1 \rangle$
from $p_1$ to $q_1$ has length less than $\dist(p_1,q_1)$.
(Notice that $\{p_1,p_2\}$ and, similarly, $\{q_1, q_2\}$, are in $E(H)$
since both $p_1$ and $p_2$ are inside a ball of radius $2^{i-2}$
and the transmission ranges of $p_1$ and $p_2$ is at least
$2^{i}$, as they are incident to edges in $E_i$.)
Since $\{p_1,q_1\}$ is a primitive edge in $H$ and $H\in\mathcal{T}(Q)$,
$\{p_1, q\}$ cannot be bridged, deriving a contradiction.

A contradiction is derived in both cases. Therefore, \eqref{eqn:log} holds.
The result follows by \eqref{eqn:transformation} and \eqref{eqn:log} 
since set $P$ and graph $G$ correspond to $Q$ and $H$, 
respectively, upon scaling by $1/\alpha = d_{\min}(G)$.
\end{proof}

In the next lemma we show that $\mst(P)'$ is in $\mathcal{T}(P)$.
Consequently, $\mathcal{T}(P)$ is always non-empty.

\begin{lemma}
\label{lem:mst}
For any set of points $P \subseteq \reals^d$, $\mst(P)'\in \mathcal{T}(P)$.
\end{lemma}

\begin{proof}
The transmission range of each node $p \in P$ is determined by the length of
the longest edge adjacent to $p$ in $\mst(P)$.
Suppose there is a primitive edge $\{p_1, p_2\}\in \mst(P)$ that is bridged.
Therefore, there is a path $T$ from $p_1$ to $p_2$ in $\mst(P)'$
that contains at most three edges, each of which is of length less than
$\dist(p_1,p_2)$.
Removing the edge $\{p_1,p_2\}$ partitions $\mst(P)$ into two connected
components, where $p_1$ and $p_2$ are in different components.
By definition, $T$ contains an edge that spans the two components.
The two components can be joined using this edge
(of length less than $\dist(p_1,p_2)$) to obtain a new spanning tree whose
weight is less than that of $\mst(P)$, deriving a contradiction.
Therefore, no primitive edge $\{p_1, p_2\}\in \mst(P)$ can be bridged,
implying $\mst(P)' \in \mathcal{T}(P)$.
\end{proof}

Theorem~\ref{thm:log} implies that the interference of any graph $G$
in $\mathcal{T}(P)$
is bounded asymptotically by the logarithm of the ratio of the longest
and shortest edges in $G$.
While this ratio can be arbitrarily large in the worst case,
we show that the ratio is bounded for many typical distributions of points.
Specifically, if the ratio is $O(n^c)$ for some constant $c$,
then the maximum interference is $O(\log n)$.

\begin{definition}[$\mathcal{D}$]%: a class of distributions]
\label{def:classD}
Let $\mathcal{D}$ denote the class of distributions over $[0,1]^d$ such that 
for any $D \in \mathcal{D}$ and any set $P$ of $n \geq 2$ points 
selected independently at random according to $D$,
the minimum distance between any two points in $P$ is greater than $n^{-c}$
with high probability, for some constant $c$ (independent of $n$).
\end{definition}

\begin{theorem}
\label{thm:main}
For any integers $d \geq 1$ and $n\geq 2$,
any distribution $D \in \mathcal{D}$,
and any set $P$ of $n$ points, each of which is selected independently
at random over $[0,1]^d$ according to distribution $D$,
with high probability, for all graphs $G\in\mathcal{T}(P)$,
$\inter(G) \in O(\log n)$.
\end{theorem}

\begin{proof}
Let $d_{\min}(G)= \min_{\{s,t\}\in E(G)} \dist(s,t)$
and $d_{\max}(G)=\max_{\{s,t\}\in E(G)} \dist(s,t)$.
Since points are contained in $[0,1]^d$, $d_{\max}(G)\leq \sqrt{d}$.
Points in $P$ are distributed according to a distribution $D \in \mathcal{D}$.
By Definition~\ref{def:classD},
with high probability, $d_{\min}(G)\geq n^{-c}$ for some constant $c$.
Thus, with high probability, we have
\begin{equation}
\label{eqn:main}
    \log \left( \frac{d_{\max}(G)}{d_{\min}(G)} \right) \leq
    \log \left( \frac{\sqrt{d}}{n^{-c}} \right).
\end{equation}
The result follows from \eqref{eqn:main}, Theorem~\ref{thm:log}, and the fact
that $\log ( n^c \sqrt{d} ) \in O(\log n)$ when $d$ and $c$ are constant.
\end{proof}

\begin{lemma}
\label{lem:bonded-Dist}
Let $D$ be a distribution with domain $[0,1]^d$,
for which there is a constant $c'$ such that for any
point $x\in [0,1]^d$, we have $D(x)\leq c'$,
where $D(x)$ denotes the probability density function of $D$ at $x \in [0,1]^d$.
Then $D\in\mathcal{D}$.
\end{lemma}

\begin{proof}
Let $p_1, p_2, \ldots, p_n$,
be $n \geq 2$ independent random points in $[0,1]^d$ with distribution $D$.
Let $c'' = 1+\frac{\log c'+2}{d}$ and let $\mathcal{E}_i$, $1\leq i\leq n$,
denote the event that there is a point $p_j$, $j\neq i$,
such that $\dist(p_i,p_j)\leq n^{-c''}$.
Let the random variable $d_{\min}$ be equal to $\min_{i\neq j}\dist(p_i, p_j)$.
  We have
  \begin{equation}
  \label{equ:union}
\hspace{-1pt}
    \pr(d_{\min}\leq n^{-c''})
=\pr\left(\bigvee_{1\leq i\leq n}\mathcal{E}_i\right)
\leq \sum_{1\leq i \leq n}\pr(\mathcal{E}_i),
\hspace{-1pt}
  \end{equation}
where the inequality holds by the union bound.
To establish an upper bound on $\pr(\mathcal{E}_i)$, 
consider a $d$-dimen\-sion\-al ball $B_i$ with centre $p_i$ 
and radius $n^{-c''}$.
The probability that there is point $p_j$, $j\neq i$, 
in that ball is at most $c'$ times the volume of $B_i\cap [0,1]^d$.
The volume of $B_i\cap [0,1]^d$ is  at most $(2n^{-c''})^d$.
Therefore, $\pr(\mathcal{E}_i)\leq c'(2n^{-c''})^d$ for every $1\leq i\leq n$.
Thus, by \eqref{equ:union}, we get
\begin{align*}
\pr(d_{\min}> n^{-c''})
& \geq 1-\sum_{1\leq i \leq n}\pr(\mathcal{E}_i)\\
&\geq 1-n\cdot c' \left(2n^{-c''}\right)^d\\
%&= 1-\frac{c'2^d}{n^{c''d-1}}\\
%&=1-\frac{c'2^d}{n^{\left(1+\frac{\log c'+2}{d}\right)d-1}}\\
&=1-\frac{c'2^d}{n^{d+\log c'+1}}\\
%&=1-\frac{c'2^d}{n\cdot n^{d+\log c'}}\\
&\geq 1-\frac{c'2^d}{n\cdot 2^{d+\log c'}}\\
&=1-\frac{1}{n}.
\end{align*}
%\begin{align*}
%\pr(d_{\min}> n^{-c''})
%& \geq 1-\sum_{1\leq i \leq n}\pr(\mathcal{E}_i)
%&&\geq 1-n\cdot c' \left(2n^{-c''}\right)^d\\
%%&= 1-\frac{c'2^d}{n^{c''d-1}}\\
%%&=1-\frac{c'2^d}{n^{\left(1+\frac{\log c'+2}{d}\right)d-1}}\\
%=1-\frac{c'2^d}{n^{d+\log c'+1}}
%%&=1-\frac{c'2^d}{n\cdot n^{d+\log c'}}\\
%&\geq 1-\frac{c'2^d}{n\cdot 2^{d+\log c'}}
%&&=1-\frac{1}{n}.
%\end{align*}
Therefore, $D\in\mathcal{D}$.
Note, here $c=c''$ in Definition~\ref{def:classD}.
\end{proof}

\begin{cor}\label{cor:uniformDist}
The uniform distribution with domain $[0,1]^d$ is in $\mathcal{D}$.
\end{cor}

By Corollary~\ref{cor:uniformDist}
and Theorem~\ref{thm:main}, we can conclude that
if a set $P$ of $n\geq 2$ points is distributed uniformly in $[0,1]^d$,
then with high probability, any communicaiton graph in
$G\in\mathcal{T}(P)$ will have maximum interference $O(\log n)$.
This is expressed formally in the following corollary:

\begin{cor}%\label{cor:main}
Choose any integers $d \geq 1$ and $n\geq 2$.
Let $P$ be a set of $n$ points, each of which is
selected independently and uniformly at random over $[0,1]^d$.
With high probability, for all graphs $G\in\mathcal{T}(P)$,
\[ \inter(G) \in O(\log n).  \]
\end{cor}

%%% SUBSECTION
\subsection{Mobility}
Our results apply to the setting of mobility
(e.g., mobile ad hoc wireless networks).
Each node in a mobile network must periodically
exchange information with its neighbours to
update its local data storing positions and transmission radii 
of nodes within its local neighbourhood.
The distribution of mobile nodes depends on the mobility model,
which is not necessarily uniform.
For example, when the network is distributed over a disc or a box-shaped region,
the probability distribution associated with the random waypoint model 
achieves its maximum at the centre of the region,
whereas the probability of finding a node close to the region's boundary 
approaches zero \cite{hyytia2006}.
Since the maximum value of the probability distribution associated with 
the random waypoint model is constant \cite{hyytia2006}, 
by Lemma~\ref{lem:bonded-Dist} and Theorem~\ref{thm:main}, 
we can conclude that at any point in time, the maximum interference 
of the network is $O(\log n)$ with high probability.
In general, this holds for any random mobility model whose corresponding 
probability distribution has a constant maximum value.

%%% SUBSECTION
\subsection{Local Algorithm}
\label{sec:localAlg}
As discussed in Section~\ref{sec:introduction.motivation},
existing local algorithms for topology control attempt to reduce 
transmission radii, but not necessarily the maximum interference. 
By Lemma~\ref{lem:mst} and Theorem~\ref{thm:main},
if $P$ is a set of $n$ points selected
according to a distribution in $\mathcal{D}$,
then with high probability $\inter(\mst(P)')\in O(\log n)$.
Unfortunately, a minimum spanning tree cannot be generated
using only local information \cite{khan2009}.
Thus, an interesting question is whether each node can assign itself
a transmission radius using only local information
such that the resulting communication graph belongs to $\mathcal{T}(P)$
while remaining connected.
We answer this question affirmatively and
present the first local algorithm ({\sc LocalRadiusReduction}),
that assigns a transmission radius to each node such that
if the initial communication graph $G_{\max}$ is connected,
then the resulting communication graph is a connected spanning subgraph
of $G_{\max}$ that belongs to $\mathcal{T}(P)$.
Consequently, the resulting topology has
maximum interference $O(\log n)$ with high probability
when nodes are selected according to any distribution in $\mathcal{D}$.
Our algorithm can be applied to any existing topology to refine it
and further reduce its maximum interference. Thus, our solution
can be used either independently, or paired with another topology control
strategy.
The algorithm consists of three phases, which we now describe.

Let $P$ be a set of $n\geq 2$ points in $\reals^d$ and
let $r_{\max}:P \to \reals^+$
be a function that returns the maximum transmission radius allowable
at each node.
Let $G_{\max}$ denote the communication graph determined by $P$ and $r_{\max}$.
Suppose $G_{\max}$ is connected.
Algorithm {\sc LocalRadiusReduction} assumes that each node is initially aware
of its maximum transmission radius,
its spatial coordinates, and its unique identifier.

The algorithm begins with a local data acquisition phase, during which
every node broadcasts its identity, maximum transmission radius,
and coordinates in a node data message.
Each message also specifies whether the data is associated with the
sender or whether it is forwarded from a neighbour.
Every node records the node data it receives
and retransmits those messages that were not previously forwarded.
Upon completing this phase, each node is aware of the
corresponding data for all nodes within its 2-hop neighbourhood.
The algorithm then proceeds to an asynchronous
transmission radius reduction phase.

Consider a node $u$ and let $f$ denote its furthest neighbour.
If $u$ and $f$ are bridged in $G_{\max}$,
then $u$ reduces its transmission radius to correspond to that of its
next-furthest neighbour $f'$, where $\dist(u,f') < \dist(u,f)$.
This process iterates until $u$ is not bridged with its
furthest neighbour within its reduced transmission radius.
We formalize the local transmission radius reduction algorithm
in the pseudocode in Table~\ref{tab:localRadiusReduction}
that computes the new transmission radius $r'(u)$ at node $u$.

%\medskip

%\noindent
\begin{table}
\centering
\begin{tabular}{@{}r@{\hspace{4pt}}l@{}}
\multicolumn{2}{l}{{\sc Algorithm LocalRadiusReduction}$(u)$} \\
1  & $radiusReductionComplete \leftarrow$ {\bf false} \\
2  & $r'(u) \leftarrow r_{\max}(u)$ \\
3  & $f \leftarrow u$ \sf // identify $u$'s furthest neighbour $f$ \\
4  & {\bf for} each $v \in \adj(u)$ \\
5  & \tp{\bf if} $\dist(u,v) > \dist(u,f)$ \\
6  & \tp\tp$f \leftarrow v$ \\
7  & {\bf while} $\lnot radiusReductionComplete$ \\
8  & \tp$radiusModified \leftarrow$ {\bf false} \\
9  & \tp{\bf if} {\sc Bridged}$(u,f)$ \\
10 & \tp\tp$radiusModified \leftarrow$ {\bf true} \\
11 & \tp\tp$f \leftarrow u$
\sf // identify next neighbour within distance $r'(u)$ \\
12 & \tp\tp{\bf for} each $v \in \adj(u)$ \\
13 & \tp\tp\tp{\bf if} $\dist(u,v) < r'(u)$ {\bf and}
	$\dist(u,v) > \dist(u,f)$ \\
14 & \tp\tp\tp\tp$f \leftarrow v$ \\
15 & \tp\tp$r'(u)\leftarrow \dist(u,f)$ \\
16 & \tp$radiusReductionComplete \leftarrow \lnot radiusModified$ \\
17 & {\bf return} $r'(u)$ \\
%\end{tabular}
%\caption{Algorithm {\sc LocalRadiusReduction}$(u)$} 
%\label{tab:localRadiusReduction}
%\end{table}
%
%%\noindent
%\begin{table}
%\begin{tabular}{@{}r@{\hspace{4pt}}l@{}}
~ \\
\multicolumn{2}{l}{{\sc Algorithm Bridged}$(a,b)$} \\
1  & $result \leftarrow$ {\bf false} \\
2  & {\bf for} each $v \in \adj(a)$ \\
3 & \tp {\bf if} $\max\{ \dist(a, v) , \dist(v,b) \} < \dist(a,b)$
{\bf and} $v \in \adj(b)$\\
4  & \tp\tp$result \leftarrow$ {\bf true} \\
5  & \tp{\bf for} each $w \in \adj(v)$ \\
6  & \tp\tp{\bf if}
$\max\{ \dist(a, v), \dist(v,w), \dist(w,b) \} < \dist(a,b)$ \\
& \tp\tp\tp\tp
{\bf and} $w \in \adj(b) $ \\
7  & \tp\tp\tp$result \leftarrow$ {\bf true} \\
8  & {\bf return} $result$ \\
\end{tabular}
%\caption{Algorithm {\sc Bridged}$(a,b)$} 
\caption{Algorithm {\sc LocalRadiusReduction}} 
%\label{tab:bridged}
\label{tab:localRadiusReduction}
\end{table}

%\medskip

Clearly, Algorithm {\sc LocalRadiusReduction}
is 2-local.
Since transmission radii are decreased monotonically (and never increased),
the while loop iterates $O(\Delta)$ times,
where $\Delta$ denotes the maximum vertex degree in $G_{\max}$.
Consequently,
since each call to the subroutine {\sc Bridged} terminates 
in $O(\Delta^2)$ time,
each node determines its reduced transmission radius $r'(u)$
in $O(\Delta^3)$ time.

After completing the transmission radius reduction phase,
the algorithm concludes with one final adjustement in the
transmission radius to remove asymmetric edges.
In this third and final phase,
each node $u$ broadcasts its reduced transmission radius $r'(u)$.
Consider the set of nodes
$\{v_1, \ldots , v_k\} \subseteq \adj(u)$
such that $\dist(u,v_i) = r'(u)$ for all $i$
(when points are in general position, $k=1$, and there is a unique such
node $v_1$).
If $r'(v_i) < r'(u)$ for all $i$, then $u$ can reduce
its transmission radius to that of its furthest neighbour with which
bidirectional communication is possible.
Specifically,
\begin{equation}
\label{eqn:phase3}
r'(u) \leftarrow
\max_{\substack{v \in \adj(u) \\ \dist(u,v) \leq \min \{ r'(u), r'(v) \} }}
\dist(u,v) .
\end{equation}
The new value of $r'(u)$ as defined in \eqref{eqn:phase3} is straightforward
to compute in $O(\Delta)$ time.

\begin{lemma}
\label{lem:algCorrectness}
The communication graph constructed by Algorithm {\sc LocalRadiusReduction}
is in $\mathcal{T}(P)$
and is connected if the initial communication graph $G_{\max}$ is connected.
\end{lemma}

\begin{proof}
Let $G_{\min}$ denote the communication graph constructed
by Algorithm {\sc LocalRadiusReduction}.
First, we prove that $G_{\min}$ is connected if $G_{\max}$ is connected.
Let
\begin{align*}
E_{\dif} = \{ & \{u,v\} \mid \{u,v\} \in E(G_{\max}) \setminus E(G_{\min})
~ \text{and $u$ and $v$} \\
& \text{belong to different connected components of $G_{\min}$}
\}
\end{align*}
Suppose that $G_{\max}$ is connected and $G_{\min}$ is not connected.
Therefore, $E_{\dif}\neq \varnothing$.
Let
\begin{equation}
\label{eqn:algCorrectness}
    \{u',v'\} \leftarrow \argmin_{\{u,v\} \in E_{\dif}} \dist(u,v) .
\end{equation}
Since $\{u',v'\}\not\in E(G_{\min})$, then we have that
either $r'(u')< \dist(u',v')$ or $r'(v')< \dist(u',v')$.
Without loss of generality, assume $r'(u')< \dist(u',v')$.
This implies that edge $\{u',v'\}$ is bridged in $G_{\max}$
since, otherwise, $u'$ could not reduce its transmission radius to less
than $\dist(u',v')$.
By Definition~\ref{def:bridge}, there is a path $T$ between
$u'$ and $v'$ in $G_{\max}$ that contains at most three edges,
each of which is of length less than $\dist(u',v')$.
Since $T$ spans two different connected components in $G_{\min}$,
there is an edge $\{u'',v''\}$ in $T$ such that $u''$ and $v''$
belong to two different connected components.
Therefore, $\{u'',v''\}\in E_{\dif}$,
as $\{u'',v''\} \in E(G_{\max})$ and $\{ u'',v''\} \not\in E(G_{\min})$.
Thus, $\dist(u'',v'') < \dist(u',v')$, contradicting \eqref{eqn:algCorrectness}.
Therefore, $G_{\min}$ is connected if and only if $G_{\max}$ is connected.

It remains to show that $G_{\min}\in \mathcal{T}(P)$.
Let $\{u,v\}$ be any primitive edge in $E(G_{\min})$.
It suffices to show that $\{u,v\}$ is not bridged in $G_{\min}$.
By Definition~\ref{def:primitiveEdge},
we have that $\dist(u,v) = \min\{ r'(u), r'(v) \}$.
Without loss of generality, assume $r'(u)=\dist(u,v)$.
The edge $\{u,v\}$ is not bridged in $G_{\max}$, otherwise the transmission
radius of $u$ could be further reduced, resulting in the removal
of $\{u,v\}$ at the end of the third phase (where asymmetric edges are removed).
Consequently, $\{u,v\}$ is not bridged in $G_{\min}$, as $G_{\min}$
is a subgraph of $G_{\max}$ and any edge that is bridged in $G_{\min}$
is also bridged in $G_{\max}$.
\end{proof}

More generally, since transmission radii are only decreased,
it can be shown that $G_{\min}$ and $G_{\max}$
have the same number of connected components by applying
Lemma~\ref{lem:algCorrectness} on every connected component of $G_{\max}$.

%%% SECTION - SIMULATION
\section{Simulation}\label{sec:simulation}
%%%%%%%%%%%%%%%%%%%%%%%%%%%%
% Bounding the Maximum Interference of Random Sensors in the Plane
% SECTION - SIMULATION
%
% file last modified on November 16, 2011
%%%%%%%%%%%%%%%%%%%%%%%%%%%%

We simulated our local interference minimization algorithm 
to evaluate its performance in static and mobile wireless networks.
In both settings, each node collects the list of its 2-hop neighbours
in two rounds, applies the algorithm to reduce its transmission radius,
and then broadcasts its computed transmission radius 
so neighbouring nodes can eliminate asymmetric edges and
possibly further reduce their transmission radii.
By the end of this stage, all asymmetric edges are removed
and no new asymmetric edges are generated.
Consequently, a node need not broadcast its transmission radius again
after it has been further reduced.

We applied two mobility models to simulate mobile networks:
random walk and random waypoint \cite{johnson1996}.
In both models each node's initial position
is a point selected uniformly at random over the simulation region.
In the random walk model, each node selects a new speed and direction
uniformly at random over $[v_{\min}, v_{\max}]$ and $[0, 2\pi)$, respectively,
at regular intervals.
When a node encounters the simulation region's boundary,
its direction is reversed (a rotation of $\pi$)
to remain within the simulation region with the same speed.
In the random waypoint model, 
each node moves along a straight trajectory with constant speed 
toward a destination point selected uniformly at random over 
$[v_{\min}, v_{\max}]$ and the simulation region, respectively.
Upon reaching its destination, the node stops for a random pause time,
after which it selects a new random destination and speed, 
and the process repeats.

%%% SUBSECTION
\subsection{Simulation Parameters}
We set the simulation region's dimensions to 1000 metres $\times$ 1000 metres.
For both static and dynamic networks, we varied the number of nodes $n$ 
from $50$ to $1000$ in increments of $50$.
We fixed the maximum transmission radius $r_{\max}$ for each network
to 100, 200, or 300 metres.
To compute the average maximum interference for static networks, 
for each $n$ and $r_{\max}$ we generated 100,000 static networks, each with 
$n$ nodes and maximum transmission radius $r_{\max}$, 
distributed uniformly at random in the simulation region.
To compute the average maximum interference for mobile networks, 
for each $n$ and $r_{\max}$ we generated 100,000 snapshots for each mobility 
model, each with $n$ nodes and maximum transmission radius $r_{\max}$.
We set the speed interval to $[0.2, 10]$ metres per second,
and the pause time interval to $[0, 10]$ seconds (in the waypoint model).
A snapshot of the network was recorded
once every second over a simulation of 100,000 seconds.

%%% SUBSECTION
\subsection{Simulation Results}

\begin{figure}
\centering
\includegraphics[width=0.5\linewidth]{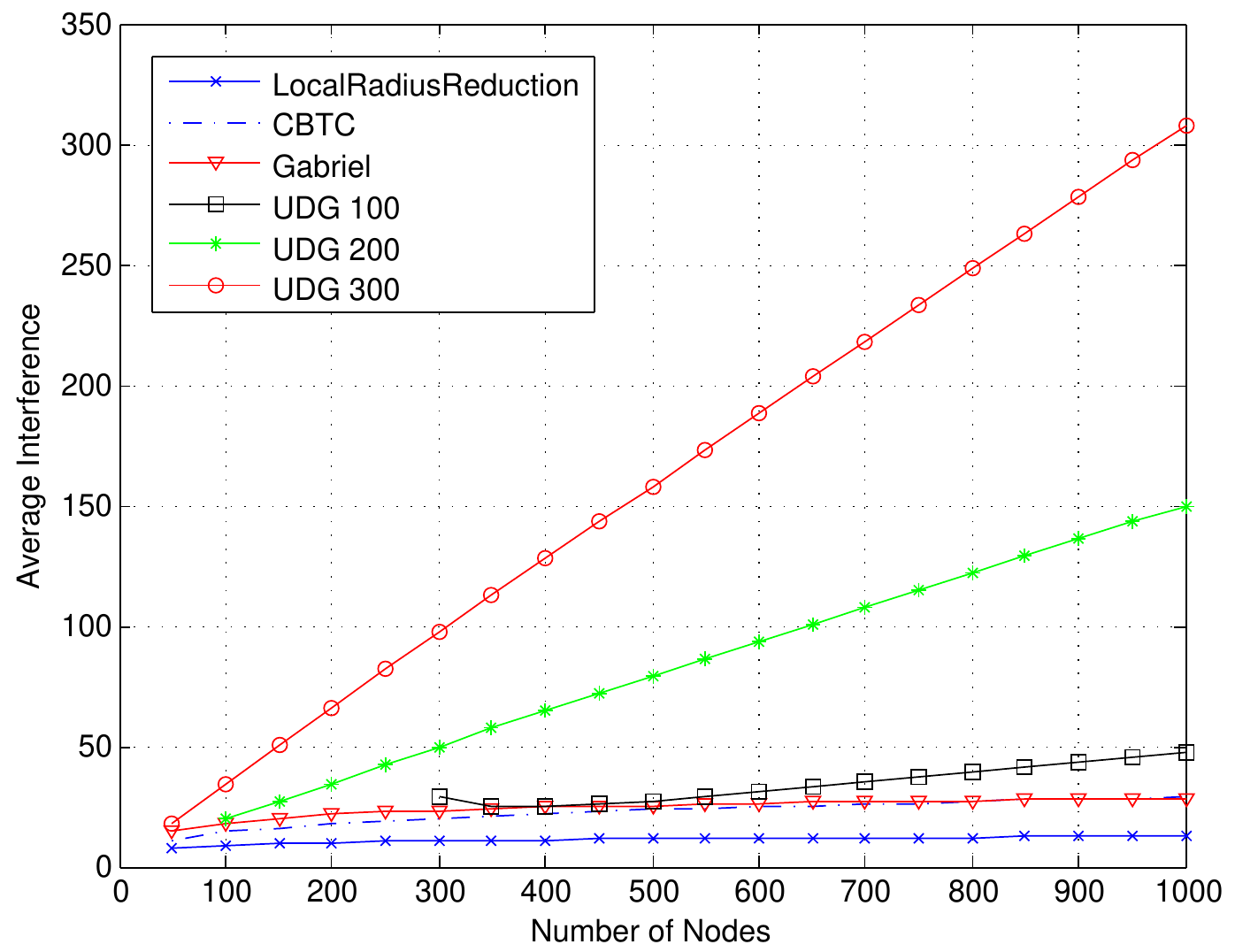}
\caption{Comparing the maximum interference of 
the {\sc LocalRadiusReduction} algorithm 
against other local topology control algorithms on a static network}
\label{fig:all_methods}
\end{figure}

\begin{figure}
\centering
\includegraphics[width=0.5\linewidth]{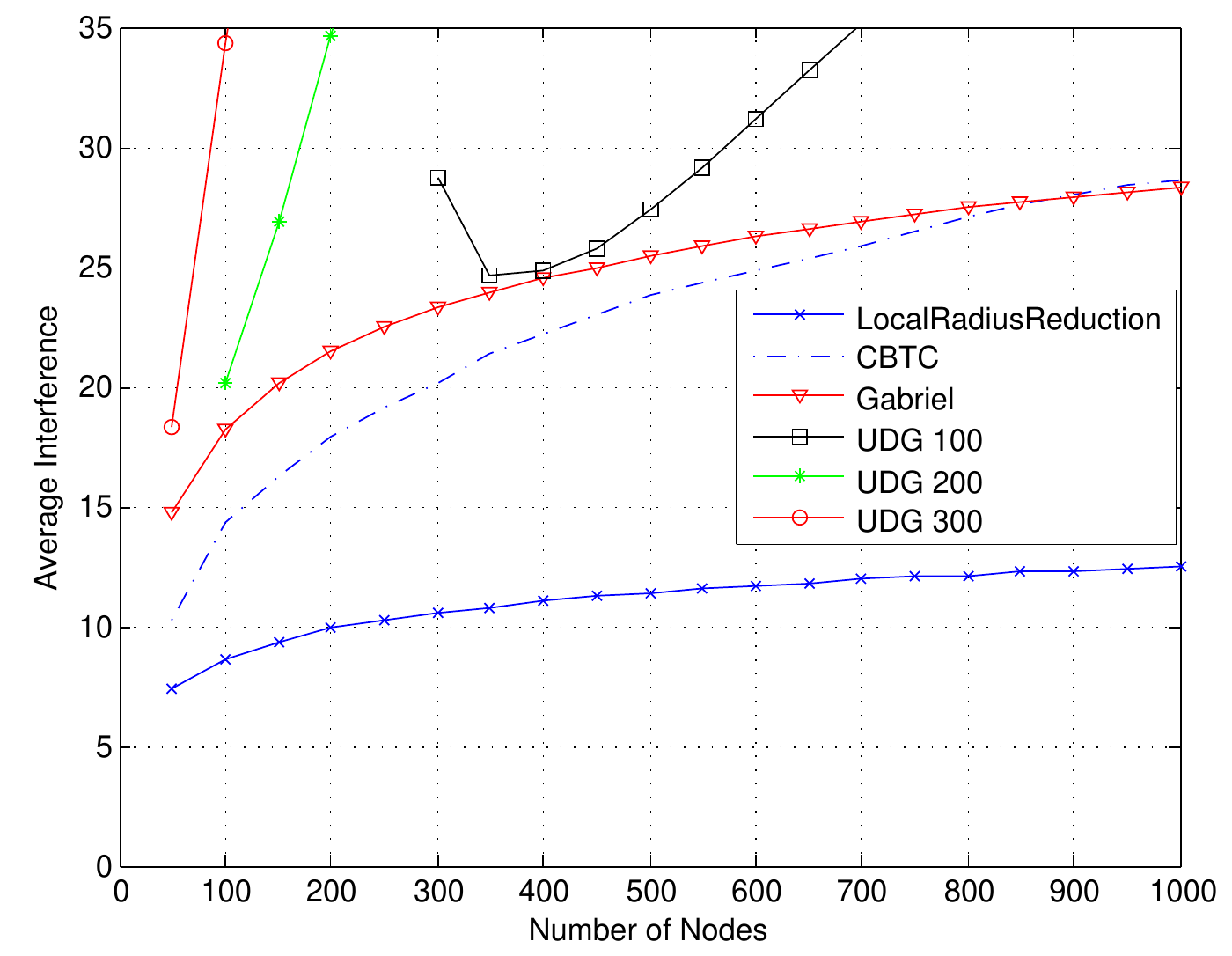}
\caption{Data from Figure~\ref{fig:all_methods} displayed with 
a bounded $y$-axis to emphasize relative differences}
\label{fig:all_method_zoom}
\end{figure}

We compared the average maximum interference of the topology constructed by
the algorithm {\sc LocalRadiusReduction} 
against the corresponding average maximum interference achieved respectively
by two local topology control algorithms: i) the local computation of 
the intersection of the Gabriel graph and the unit disc graph 
(with unit radius $r_{\max}$) \cite{boseWN2001},
and ii) the cone-based local topology control (CBTC) algorithm 
\cite{wattenhoferTON2005}.
In addition, we evaluated the maximum interference achieved when each node
uses a fixed radius of communication,
i.e., the communication graph is a unit disc graph of radius $r_{\max}$
(100, 200, or 300 metres, respectively).
These results are displayed in Figures~\ref{fig:all_methods}
and~\ref{fig:all_method_zoom}.

\begin{figure}[h]
\centering
\includegraphics[width=0.5\linewidth]{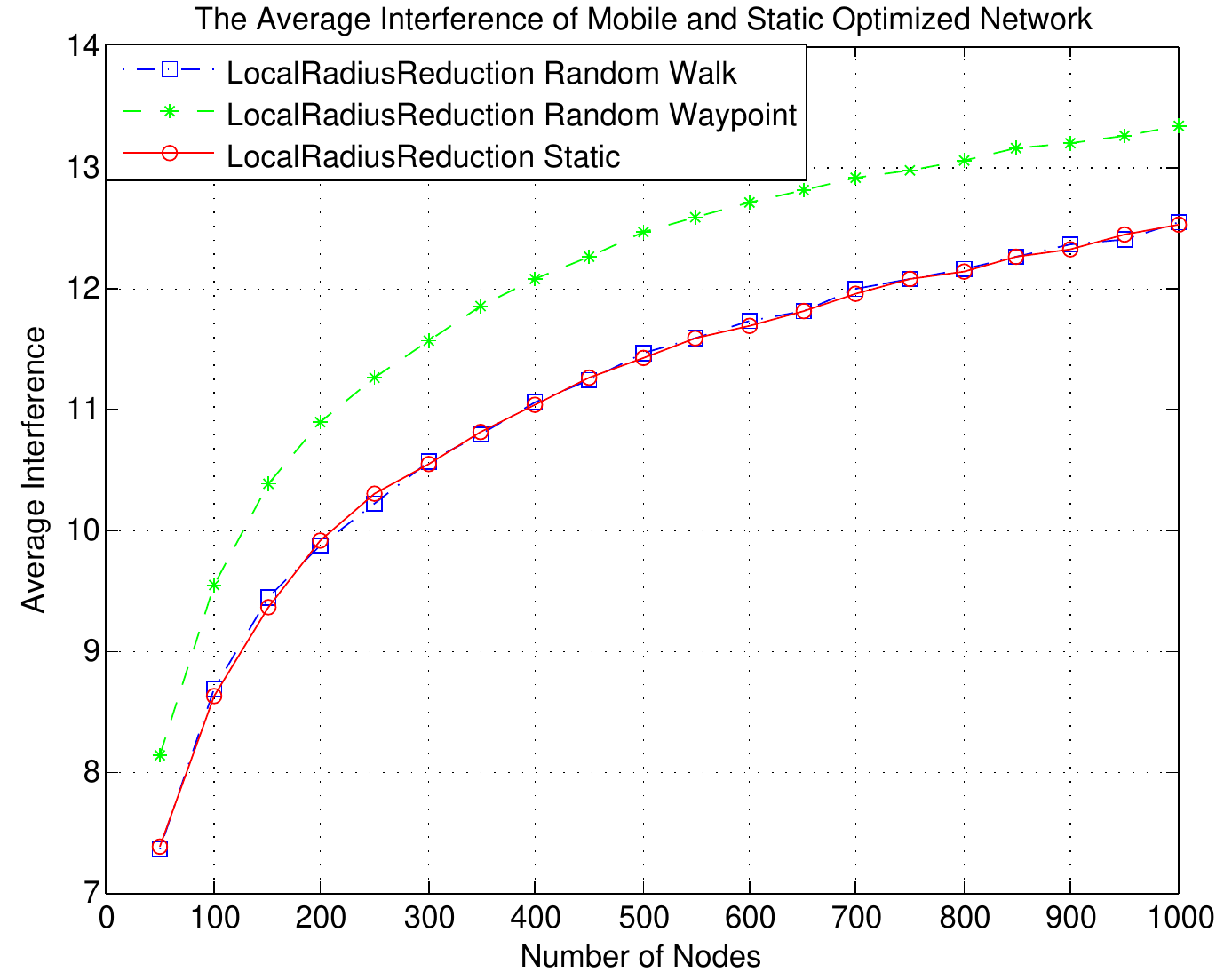}
\caption{Comparing the maximum interference of the 
{\sc LocalRadiusReduction} algorithm on static 
and mobile networks using both the random walk 
and random waypoint mobility models}
\label{fig:static_and_mobile}
\end{figure}

\begin{figure}[h]
\centering
\includegraphics[width=0.5\linewidth]{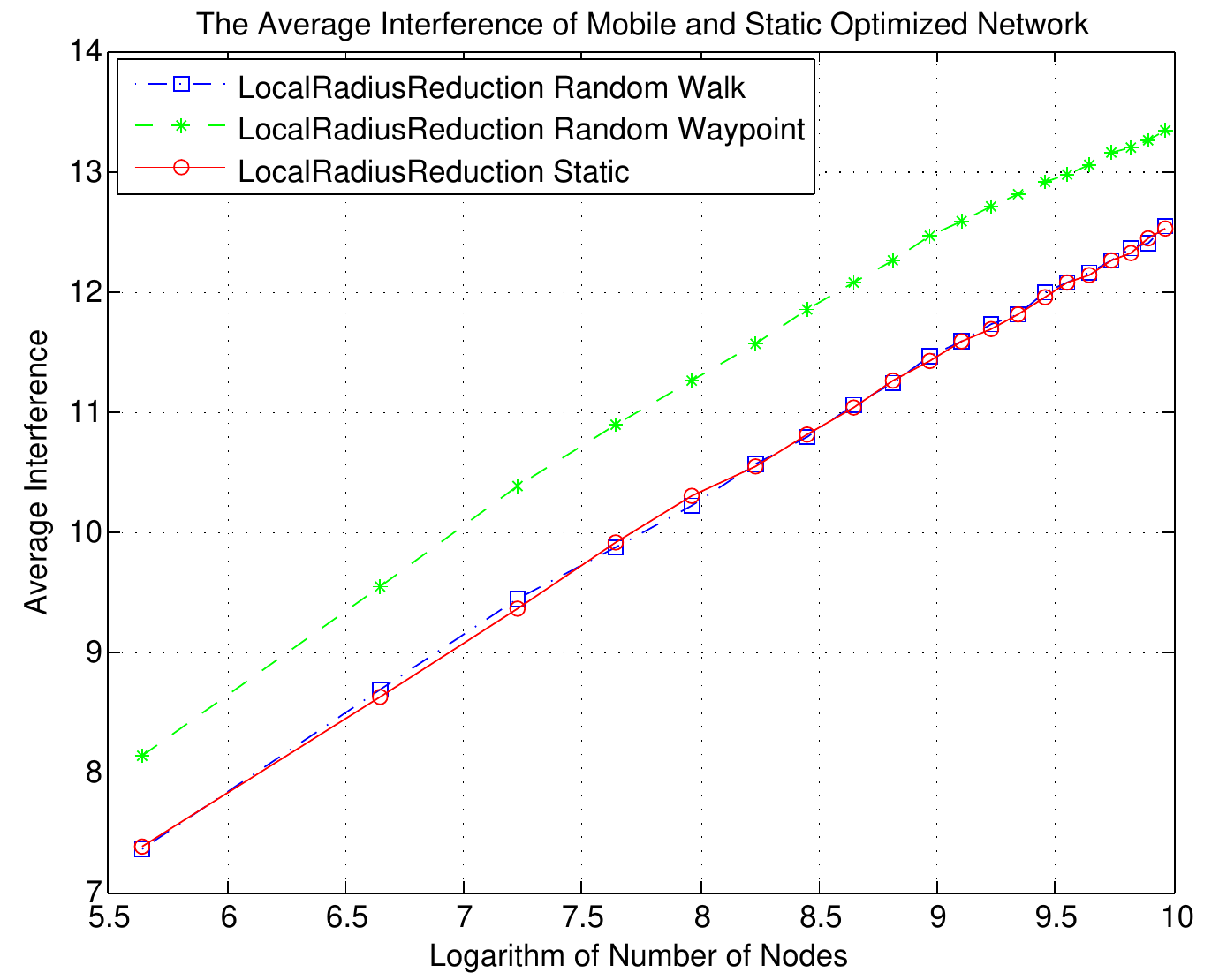}
\caption{Data from Figure~\ref{fig:static_and_mobile} displayed using
a logarithmic scale on the $x$-axis}
\label{fig:static_and_mobile_log}
\end{figure}

As shown, the average maximum inteference of the unit disc graph topologies
increases linearly with $n$.
Many of the unit disc graphs generated were disconnected
when the transmission radius was set to 100 metres for small $n$.
Since we require connectivity, we only considered values of $n$ and $r_{\max}$ 
for which at least half of the networks generated were connected.
When $r_{\max}=$ 100 metres, a higher average maximum interference 
was measured at $n=300$ than at $n=400$.
This is because many networks generated for $n=300$ were discarded due 
to being disconnected.
Consequently, the density of networks simulated for $n=300$ 
was higher than the average density of a random network with $n=300$ nodes, 
resulting in higher maximum interference.

Although both the local Gabriel and CBTC algorithms performed significantly
better than the unit disc graphs, 
the lowest average maximum interference was achieved 
by the {\sc LocalRadiusReduction} algorithm,
which is clearly seen to be logarithmic in $n$
in Figures~\ref{fig:static_and_mobile}
and \ref{fig:static_and_mobile_log}.
Note that the {\sc LocalRadiusReduction} algorithm reduces the maximum 
interference to $O(\log n)$ with high probability, 
irrespective of the initial maximum transmission radius $r_{\max}$.

Figures~\ref{fig:static_and_mobile}
and \ref{fig:static_and_mobile_log}
display the average maximum interference 
achieved by {\sc LocalRadiusReduction} 
on mobile networks, plotting simulation results for both the
random walk and random waypoint models, 
along with the corresponding results on a static network.
Simulation results obtained using the random walk model 
closely match those obtained on a static network
because the distribution of nodes
at any time during a random walk is nearly uniform \cite{dassarma2009}.
The average maximum interference increases slightly but remains logarithmic
when the random waypoint model is used.
The spatial distribution of nodes moving according to a random waypoint model 
is not uniform, and is maximized at the centre of the simulation region 
\cite{hyytia2006}.
Consequently, the density of nodes is high near the centre,
resulting in greater interference at these nodes.

Finally, we evaluated the algorithm {\sc LocalRadiusReduction} 
using actual mobility trace data of 
Piorkowski et al.\ \cite{epfl-mobility-2009-02-24},
consisting of GPS coordinates for trajectories of 537 taxi vehicles
recorded between May 17 and June 10, 2008, 
driving throughout the San Fransisco Bay area.
Each taxi's trace contains between 1000 and 20,000 sample points. 
We selected the 500 largest traces, each of which has over 8000 sample points.
To implement our algorithm, we selected $n$ taxis among the 500 uniformly
at random, ranging from $n = 50$ to $n=500$ in increments of 50.
As seen in Figure~\ref{fig:real_gps}, the resulting average maximum
interference is similar to that measured in our simulation results.

\begin{figure}[h]
\centering
\includegraphics[width=0.5\linewidth]{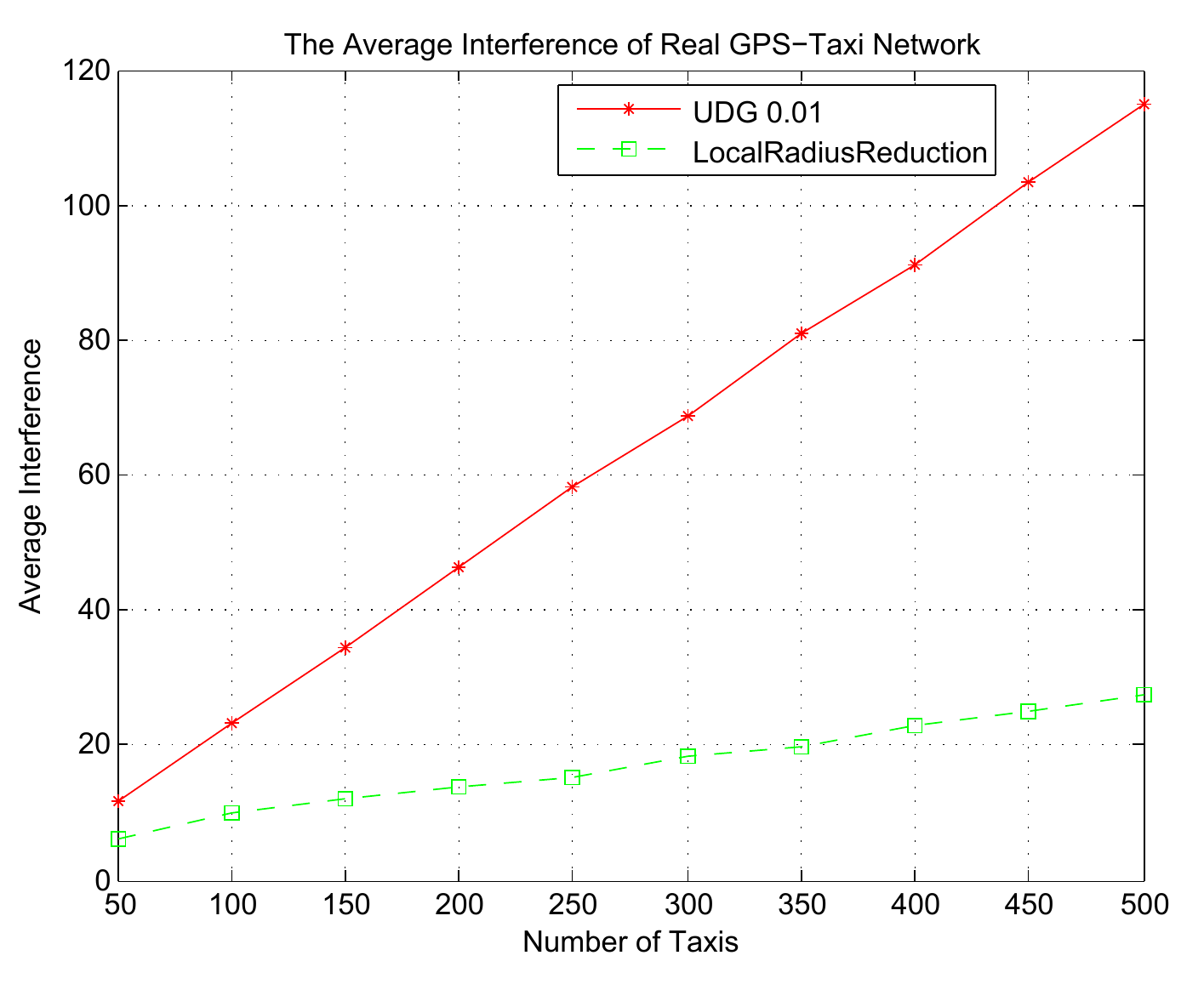}
\caption{Comparing the maximum interference of the 
{\sc LocalRadiusReduction} against a unit disc graph 
on actual mobile data}
\label{fig:real_gps}
\end{figure}

%%% SECTION - DISCUSSION
%\section{Discussion}\label{sec:conclusion}
%\input{discussion}

%%% SUBSECTION
\subsection*{Acknowledgements}
Stephane Durocher would like to thank Csaba T\'oth for insightful discussions
related to the interference minimization problem in one dimension.

%\newpage

%%% BIBLIOGRAPHY
%\bibliographystyle{plain}
%\bibliography{interferenceMin}

\end{document}